\documentclass[11pt]{article}
\usepackage[layoutsize={8.5in,11in},margin=1in]{geometry}\usepackage{graphicx}\usepackage{amsmath}\usepackage{amssymb}\usepackage{amsthm}\usepackage{mathrsfs}\usepackage{bm}\usepackage{pifont}\usepackage[dvipsnames,svgnames,x11names]{xcolor}\usepackage{setspace}\usepackage{pdfcomment}\usepackage{indentfirst}\usepackage{float}\usepackage{braket}\usepackage{enumerate}\usepackage{cleveref}\usepackage{tikz}\tikzset{vertex/.style={circle,draw=black,thick}}\tikzset{arrow/.style={-Stealth,thick}}\usetikzlibrary{decorations.pathmorphing,decorations.text,shapes.geometric,3d,positioning,arrows.meta}
\def\({\left(}\def\){\right)}\def\[{\left[}\def\]{\right]}\def\L{\left\{}\def\R{\right\}}\def\|{\middle\vert}\def\c{\cdot}\def\cs{\cdots}\def\<={\leq}\def\>={\geq}\def\f{\frac}\def\mrm{\mathrm}\def\bM{\begin{pmatrix}}\def\eM{\end{pmatrix}}\def\la{\left|}\def\ra{\right|}\def\eps{\varepsilon}\def\|{\middle\vert}\def\b{\mathbf}\def\msf{\mathsf}\def\wtilde{\widetilde}\def\C{\binom}\def\lf{\left\lfloor}\def\rf{\right\rfloor}\def\lF{\left\lceil}\def\rF{\right\rceil}\def\phantomeq{\mathrel{\hphantom{=}}}
\def\approxcount{\msf{ApproxCount}}\newcommand{\approxcountc}[1]{\approxcount_{#1\text{\rm\sf-counter}}}\newcommand{\approxcountp}[1]{\approxcount_{#1\text{\rm\sf-parallel}}}
\newcommand{\Rmnum}[1]{\uppercase\expandafter{\romannumeral#1}}\newcommand{\tops}[2]{\texorpdfstring{#1}{#2}}
\newtheorem{theorem}{Theorem}[section]\newtheorem{lemma}[theorem]{Lemma}\newtheorem{definition}[theorem]{Definition}\newtheorem{proposition}[theorem]{Proposition}\newtheorem{corollary}[theorem]{Corollary}\newtheorem{openproblem}[theorem]{Open Problem}\newenvironment{Proof}{\begin{proof}~\par}{\end{proof}}
\mathchardef\standardl=\mathcode`l \mathcode`l=\ell \newcommand{\deactivatel}{\mathcode`l=\standardl} \makeatletter \edef\operator@font{\operator@font\noexpand\deactivatel} \makeatother
\ExplSyntaxOn\NewDocumentCommand{\spreadwords}{m}{\tl_set:Nn\l_tmpa_tl{#1}\tl_replace_all:Nnn\l_tmpa_tl{~}{\hspace{\stretch{1}}}\tl_use:N\l_tmpa_tl}\ExplSyntaxOff

\title{Tight Streaming Lower Bounds for Deterministic Approximate Counting}
\author{
	Yichuan Wang\footnote{Institute for Interdisciplinary Information Sciences, Tsinghua University. {\tt yichuan-21@mails.tsinghua.edu.cn}.}
}
\date{2024-06-17}

\begin{document}

\begin{titlepage}
	\maketitle
	\thispagestyle{empty}
	\begin{abstract}\normalsize
		We study the streaming complexity of {\it $k$-counter approximate counting}. In the $k$-counter approximate counting problem, we are given an input string in $[k]^n$, and we are required to approximate the number of each $j$'s ($j\in[k]$) in the string. Typically we require an additive error $\<=\f{n}{3(k-1)}$ for each $j\in[k]$ respectively, and we are mostly interested in the regime $n\gg k$. We prove a lower bound result that the deterministic and worst-case {\it $k$-counter approximate counting} problem requires $\Omega(k\log(n/k))$ bits of space in the streaming model, while no non-trivial lower bounds were known before. In contrast, trivially counting the number of each $j\in[k]$ uses $O(k\log n)$ bits of space. Our main proof technique is analyzing a novel potential function.\par
		Our lower bound for $k$-counter approximate counting also implies the optimality of some other streaming algorithms. For example, we show that the celebrated Misra-Gries algorithm for heavy hitters \cite{MG82} has achieved optimal space usage.
	\end{abstract}
\end{titlepage}

\thispagestyle{empty}
\tableofcontents
\newpage
\setcounter{page}{1}

\section{Introduction}\label{sec.intro}

We study the {\it $k$-counter approximate counting} problem. Given an input string $\b{x}\in[k]^n$, for each $j\in[k]$, define the {\it frequency} of $j$ in $\b{x}$ as the number of $j$'s in $\b{x}$. The $k$-counter approximate counting problem asks us to approximate the frequency of each $j\in[k]$ in the input string. More precisely, define $\approxcountc k[n,\Delta]$ as the problem that asks us to approximate the frequency of each $j\in[k]$ in an input string in $[k]^n$, with an additive error $\<=\Delta$ for each $j\in[k]$ respectively.\par
The $k=2$ case is equivalent to estimating the frequency of $1$ (i.e., the Hamming weight) in a $0,1$-string, we denote this problem by {\it $\{0,1\}$-approximate counting}, and by $\approxcount[n,\Delta]$ if we require the additive error to be $\<=\Delta$. $\{0,1\}$-approximate counting appears as a gadget in many more sophisticated algorithms, and thus studying its complexity is a fundamental question in theoretical computer science. The general $k$-counter approximate counting problem also has its own interest since it captures some scenarios in which we wish to aggregate over a large dataset, and each data point is a general word (instead of a bit).\par
In this work we only focus on the deterministic and worst-case approximate counting.\par
In the streaming model, trivially counting the exact frequency of each element in a string in $[k]^n$ uses $O(k\log n)$ bits of space. In the regime $n\gg k$,\footnote{Throughout this paper, $f\gg g$ means $f\>=\omega(g)$.} previously we did not know whether we can use less space even if we only need to output an approximation of the frequency. Our work shows the trivial algorithm is already optimal in some sense.

\begin{theorem}{\rm (Direct from \Cref{intro.thm.robp-lb})}\label{intro.thm.main}
	For any integers $n,k$ such that $k\>=2$, $n\>=3k$, computing $\approxcountc k\[n,\f{n}{3(k-1)}\]$ requires $\Omega(k\log(n/k))$ bits of space in the streaming model.
\end{theorem}

\subsection{Implications: Other Streaming Lower Bounds}

Since $k$-counter approximate counting is a fundamental problem, our lower bound also implies streaming lower bounds for some other problems.

\subsubsection*{Lower Bound for Heavy Hitters}

The seminal paper \cite{MG82} presented a deterministic streaming algorithm (the Misra-Gries algorithm) that solves the following problem: Given an input string $\b{x}\in[U]^n$. Let $k\in\mathbb{Z}^+$ be some parameter such that $\min\{n,U\}\gg k\>=2$. Their algorithm outputs a list of $k$ elements $u_1,\cs,u_k\in[U]$ such that each element in $[U]$ that appears $\>=n/k$ times in $\b{x}$ is contained in $\{u_1,\cs,u_k\}$. Moreover, let $f_i$ be the frequency of $u_i$ in $\b{x}$, then their algorithm also outputs estimates $\wtilde{f_1},\cs,\wtilde{f_k}$ of $f_1,\cs,f_k$ such that $f_i-n/k\<=\wtilde{f_i}\<=f_i$. Their algorithm uses $O(k(\log(n/k)+\log(U/k)))$ bits of space.\par
The problem they solved was later named the {\it heavy hitters} problem. It is easy to observe that the $\Omega(k\log(U/k))$ bits of space usage is necessary\footnote{Since if exactly $k$ elements appear $n/k$ times each, then the algorithm need to have $\>=\C{U}{k}\>=2^{\Omega(k\log(U/k))}$ many possible outputs.}. However, previously we did not know whether the $\Omega(k\log(n/k))$ bits of space usage is necessary in the regime $n\gg U$.\par
Based on our lower bounds for approximate counting, we show the $\Omega(k\log(n/k))$ bits of space usage is necessary.

\begin{theorem}{\rm (See also \Cref{thm.lb-heavyhitters})}\label{intro.thm.lb-heavyhitters}
	Computing the heavy hitters problem described above\footnote{Here we are required to output both the list $\{u_1,\cs,u_k\}$ and the estimates $\wtilde{f_1},\cs,\wtilde{f_k}$. If we are only required to output the list $\{u_1,\cs,u_k\}$, then the lower bound is still open. See \Cref{sec.imply} for more discussions.} requires $\Omega\(k(\log(n/k)+\log(U/k))\)$ bits of space in the streaming model.
\end{theorem}

\subsubsection*{Lower Bound for Quantile Sketch}

The {\it quantile sketch} problem is another interesting problem in streaming algorithms. We first define the {\it rank} of an element $u$ in a list $L$ as the number of elements in $L$ that are at most $u$. In the quantile sketch problem, we are given an input string $(x_1,\cs,x_n)\in[U]^n$. For any $1\<=t\<=n$, after reading $x_1,\cs,x_t$, we are required to answer some queries that ask us to approximate the rank of some element $u$ in $x_1,\cs,x_t$, with an additive error $\<=\eps t$. Here $\eps\in(0,1)$ is some parameter. We are mostly interested in the regime $\min\{n,U\}\gg1/\eps\gg 1$.\par
In a recent breakthrough, \cite{GSW24} proposed a deterministic streaming algorithm for the quantile sketch problem that uses $O\(\eps^{-1}(\log(\eps U)+\log(\eps n))\)$ bits of space. They also proved a lower bound that quantile sketch requires $\Omega\(\eps^{-1}(\log(\eps U)+\log(\eps n))\)$ bits of space, under a conjecture on the hardness of approximate counting, see \Cref{sec.imply} for more details. Our lower bounds for approximate counting resolve their conjecture, and therefore their algorithm is optimal.

\begin{theorem}{\rm (See also \Cref{thm.lb-quantile})}\label{intro.thm.lb-quantile}
	Computing the quantile sketch problem described above requires $\Omega\(\eps^{-1}(\log(\eps U)+\log(\eps n))\)$ bits of space in the streaming model.
\end{theorem}

We mention that the lower bounds for heavy hitters and quantile sketch are only against deterministic and worst-case streaming algorithms.

\subsection{Model: Read-once Branching Programs (ROBP)}

We use the read-once branching program (ROBP) model, which is the standard model to study streaming lower bounds. We give an intuitive description of ROBP here. (See \Cref{def.ROBP} for a formal definition.) An ROBP is a layered multigraph. The vertices in each layer characterize all possible memory states of a streaming program. Each vertex (except those in the last layer) has outgoing edges labeled with possible input words, which indicates the next state after reading an input word.\par
ROBP lower bounds imply streaming lower bounds. Define the width of an ROBP as the maximal number of vertices in any layer. It is easy to observe that if a problem requires width $\>=w$ in ROBP, then it requires $\Omega(\log w)$ bits of space in the streaming model. We need to mention that by ROBP we do not care about the computational resources needed to compute the next state from the previous state and the input word. However, the ROBP model is still sufficient for us to establish the lower bound results.

\subsection{Tighter Bounds on the Width of the ROBP}

Actually, our proof techniques can give a much tighter bound on the width of the ROBP.

\begin{theorem}{\rm (See also \Cref{cor.k-counter-lb.standard})}\label{intro.thm.robp-lb}
	For any integers $n,k$ such that $k\>=2$, $n\>=3k$, computing $\approxcountc k\[n,\f{n}{3(k-1)}\]$ requires $\Omega(n/k)^{k-1}$ width in the ROBP model.
\end{theorem}

In contrast, a weaker $(n/k)^{\Omega(k)}$ lower bound on the width is sufficient to imply \Cref{intro.thm.main}. Also note that the trivial algorithm for exact counting uses $\C{n+k-1}{k-1}\<=O(n/k)^{k-1}$ width.\par
We need to mention that, even for the $k=2$ case (which is equivalent to $\{0,1\}$-approximate counting), the $\Omega(n)$ lower bound in \Cref{intro.thm.robp-lb} was not known before. We did not even know whether $\approxcount[n,n/3]$ can be computed by some width-$10$ ROBP. Previously we only knew that $\{0,1\}$-approximate counting with a constant multiplicative error requires $n^{\Omega(1)}$ width, which was proved in \cite{ABJ+22}.\par
We also discover some non-trivial ROBP algorithms for approximate counting. Let's take $\{0,1\}$-approximate counting for an example. At first glance, the optimal algorithm we can think of is to count the exact number of $1$'s in the first $w-1$ input bits and ignore the other bits. This achieves an additive error $\<=(n-w+1)/2$ and uses width-$w$. However, it turns out that there exists a simple algorithm that achieves an additive error $\<=n/2-\Omega(\sqrt{n})$ and uses width-$3$,\footnote{See \Cref{thm.algo.small-w} for details.} which is slightly better than the trivial one. More surprisingly, these algorithms almost match our lower bounds in some regimes in a very tight sense. 

\begin{theorem}{\rm (See also \Cref{cor.tight-small-w})}\label{intro.thm.tight-small-w}
	For any integers $n\>=1$ and $3\<=w\<=n/10$, let $\Delta(n,w)$ be the minimal value of $\Delta$ such that $\approxcount[n,\Delta]$ can be computed by a width-$w$ ROBP. Then $\Delta(n,w)\in n/2-\Theta(\sqrt{nw})$.
\end{theorem}

\begin{theorem}{\rm (See also \Cref{cor.tight-small-err})}\label{intro.thm.tight-small-err}
	For any integers $k\>=2$, $n\>=10k$ and real number $10\<=\Delta\<=n/(10k^2)$, let $w(k,n,\Delta)$ be the minimal value of $w$ such that $\approxcountc k[n,\Delta]$ can be computed by a width-$w$ ROBP. Then $w(k,n,\Delta)\in\C{n+k-1-\Theta(\sqrt{n\Delta})}{k-1}$.
\end{theorem}

The existence of the non-trivial algorithms also shows proving the lower bounds is hard. Without these non-trivial algorithms, we may hope to prove that the simple and trivial algorithm (which achieves $(n-w+1)/2$ additive error) is already optimal. However the non-trivial algorithms have ruled out this hope.\par
We also consider the {\it$k$-parallel approximate counting} problem, in which we are required to solve $k$ instances of $\{0,1\}$-approximate counting that arrive in parallel (i.e., a {\it direct sum}). More precisely, given an input string $(x_1,\cs,x_n)\in\(\{0,1\}^k\)^n$, the problem $\approxcountp k[n,\Delta]$ asks us to approximate the number of $1$'s in $(x_1)_j,\cs,(x_n)_j$ with an additive error $\<=\Delta$, for each $j\in[k]$ respectively. We give the following lower bound.

\begin{theorem}{\rm (See also \Cref{thm.k-par-lb})}\label{intro.thm.k-par}
	For any integers $n,k$ such that $k\>=1$, $n\>=3k$, computing $\approxcountp k[n,n/3]$ requires $n^{\Omega(k)}$ width in the ROBP model.
\end{theorem}

However, our current proof techniques cannot show an $\Omega(n)^k$ lower bound on the width (as in \Cref{intro.thm.robp-lb}). See \Cref{sec.k-par-lb} for more discussion.

\begin{openproblem}{\rm (See also \Cref{open.k-par})}\label{intro.open}
	Prove or disprove: for any integers $n,k$ such that $n\gg k\>= 1$, computing $\approxcountp k[n,n/3]$ requires $\Omega(n)^k$ width in the ROBP model.
\end{openproblem}

Note that \Cref{intro.thm.k-par} can be viewed as a {\it direct sum theorem} for $\{0,1\}$-approximate counting in the ROBP model. i.e., to compute on $k$ instances that arrive in parallel, independently computing on each instance is almost optimal. In contrast, for general problems, we only know a much weaker direct sum theorem that if computing some function $f$ (on $n$ bits) requires average space\footnote{By average space we mean the average of number of memory bits over all layers.} $s$, then computing the $k$-parallel version of $f$ requires average space $\Omega(ks/n)$ \cite{RS16}, which is not sufficient for us to establish \Cref{intro.thm.k-par}.

\subsection{More on the Complexity of Approximate Counting}

It was already known that approximate counting is easy in some other computation models, which are slightly stronger than deterministic ROBP but still very weak. Therefore our lower bounds shed light on what is the exact weakest model to compute approximate counting.\par
If we relax {\it deterministic} streaming to {\it randomized} streaming, then approximate counting becomes very easy. Since we can only sample a few words of the input and count the fraction of each word in the sampled part. By concentration inequalities (e.g., the Chernoff Bound), we can show that this randomized algorithm achieves a small error with high probability.\par
If we consider approximate counting with a small multiplicative error (instead of additive error in our work), \cite{Mor78} presented a super efficient randomized algorithm, which was called {\it Morris' Counter}. \cite{NY22} improved that algorithm and also gave a matching lower bound. Their optimal algorithm uses $O(\min\{\log n,\log\log n+\log(1/\eps)+\log\log(1/\delta)\})$ bits of space to estimate the number of $1$'s in an $n$-bit $0,1$-string up to a multiplicative error of $\eps$, with success probability $1-\delta$. Their lower bound was also generalized to the $k$-counter setting in \cite{AHNY22}.\par
Another surprising algorithm for approximate counting is that (deterministic and worst-case) $\approxcount[n,n/3]$ can be computed in uniform $\msf{AC}^0$ \cite{Ajt90}, which is a very weak complexity class.\par
The {\it coin problem} is a problem very similar to approximate counting, which has received a lot of attention in recent literature. In the coin problem, there is a hidden probability $p\in[0,1]$, and we are given $n$ i.i.d. random bits, each of which has probability $p$ to be $1$. Our goal is to estimate the value of $p$. The streaming complexity of the coin problem was studied in \cite{BGW20} \cite{BGZ21}.

\subsection{Organization of this Paper}

In \Cref{sec.overview}, we give a high-level overview of our main proof techniques. In \Cref{sec.prelim}, we give the formal definition of the notions. In \Cref{sec.k-counter-lb}, we present the proof of our main lower bound result for $k$-counter approximate counting. In \Cref{sec.k-par-lb}, we generalize the proof technique to prove a lower bound for $k$-parallel approximate counting. In \Cref{sec.algo}, we present some non-trivial algorithms for approximate counting. In \Cref{sec.imply}, we prove streaming lower bounds for other problems, based on our lower bound for approximate counting.

\section{Proof Overview}\label{sec.overview}

In this section, we give a high-level overview of the main proof techniques we use in our proof of the lower bound for $k$-counter approximate counting.\par
Let's first focus on the $k=2$ case, or equivalently, $\{0,1\}$-approximate counting, which is already non-trivial.

\subsection{Known Techniques do not Work}

Before presenting the overview of our proof, we first need to illustrate why proving our lower bound results is difficult.\par
A popular approach for streaming lower bounds is the {\it communication bottleneck} method. Consider dividing the input string into two halves. Then for any width-$w$ ROBP, it can only remember $\log w$ bits of information about the first half before moving on to the second half. This motivates us to consider the following communication problem: Alice is given the first half of the input and Bob is given the second half, and Alice can send only one message to Bob, and Bob needs to output the answer. If we can prove a communication lower bound on the length of Alice's message, then it immediately gives a streaming lower bound. We refer to \cite{Rou16} for some examples.\par
However, note that in order for Alice and Bob to compute, say, $\approxcount[n,n/10]$, it suffices for Alice to only send an estimate of the number of $1$'s in the first half to Bob, with an additive error $\<=n/20$. This only needs $O(1)$ bits of communication. Thus the {\it communication bottleneck} method does not even seem to be able to prove that $\approxcount[n,n/10]$ requires $\omega(1)$ width.

\subsection{Lower Bound for \tops{$\{0,1\}$}{\{0,1\}}-Approximate Counting}

We first sketch how to prove $\approxcount[n,n/10]$ requires $\Omega(n)$ width in ROBP.\par
Let's first consider this question: given an ROBP, how can we decide whether it correctly computes $\approxcount[n,n/10]$? The answer is to use dynamic programming. For a vertex $v$ in the $t$-th layer\footnote{Before reaching a vertex in the $t$-th layer, the ROBP has read $t$ input words.} of the ROBP, we wish to compute an interval label $[a_v,b_v]$, here $a_v,b_v$ are the minimal/maximal possible number of $1$'s in the first $t$ input bits to reach $v$. Once we have computed all interval labels in the $t$-th layer, the intervals in the next layer can be derived according to the following rule: if the outgoing edge with label $z\in\{0,1\}$ of some vertex $u$  points to $v$, then $[a_u+z,b_u+z]$ should be a subset of $[a_v,b_v]$. See \Cref{fig.dp} for an example.

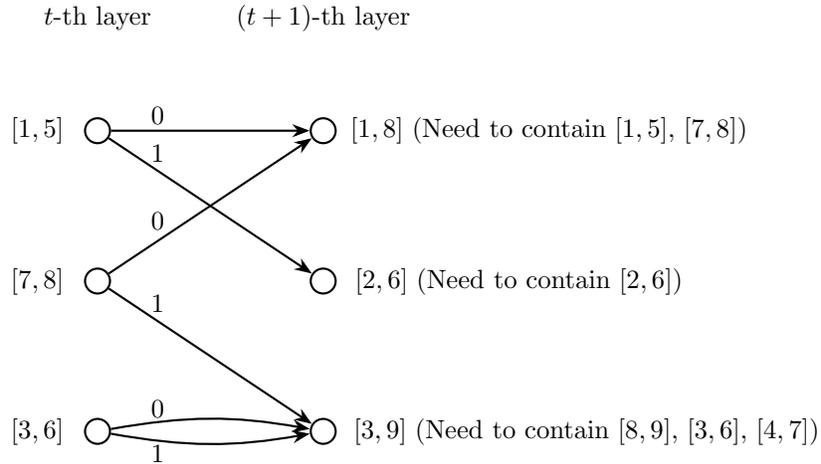
\begin{figure}[htpb]
	\centering\small
	\begin{tikzpicture}
		\node[vertex] (1) at (0,4) {};\node at (-0.8,4) {$[1,5]$};
		\node[vertex] (2) at (0,2) {};\node at (-0.8,2) {$[7,8]$};
		\node[vertex] (3) at (0,0) {};\node at (-0.8,0) {$[3,6]$};
		\node[vertex] (4) at (3,4) {};\node at (6,4) {$[1,8]$ (Need to contain $[1,5]$, $[7,8]$)};
		\node[vertex] (5) at (3,2) {};\node at (5.6,2) {$[2,6]$ (Need to contain $[2,6]$)};
		\node[vertex] (6) at (3,0) {};\node at (6.5,0) {$[3,9]$ (Need to contain $[8,9]$, $[3,6]$, $[4,7]$)};
		\draw[arrow] (1)--(4);\node at (0.8,4.2) {$0$};
		\draw[arrow] (1)--(5);\node at (0.8,3.7) {$1$};
		\draw[arrow] (2)--(4);\node at (0.8,2.8) {$0$};
		\draw[arrow] (2)--(6);\node at (0.8,1.7) {$1$};
		\draw[arrow] (3) to[out=10,in=170] (6);\node at (0.8,0.3) {$0$};
		\draw[arrow] (3) to[out=-10,in=-170] (6);\node at (0.8,-0.3) {$1$};
		\node at (0,5.5) {$t$-th layer};
		\node at (3,5.5) {$(t+1)$-th layer};
	\end{tikzpicture}
	\caption{\small An example for the dynamic programming process.}
	\label{fig.dp}
\end{figure}
We need to mention that it is not necessary that for each integer $f\in[a_v,b_v]$, there exists some string with $f$ $1$'s that reaches $v$.\par
Finally for any vertex $v$ in the last layer, we need to have $b_v-a_v\<=n/5$. This is because if we output $c_v$ as an estimate of the Hamming weight when we reach $v$, then since there are possibly $a_v$ or $b_v$ many $1$'s when we reach $v$, we should have $c_v-a_v\<=n/10$ and $b_v-c_v\<=n/10$, and hence $b_v-a_v\<=n/5$.\par
We remark that the idea of the interval labeling has already appeared in \cite{ABJ+22}'s proof of the lower bound for approximate counting with small multiplicative errors. Also note that the interval labeling strongly depends on the worst-case behavior of the ROBP, and hence it does not work in the randomized or average-case settings.\par
Now let's move on to the main step: analyzing a novel potential function. For each $0\<=t\<=n$, we wish to define a potential value $\Phi_t$ that depends on the intervals in the $t$-th layer of the ROBP. Intuitively, we wish:
\begin{enumerate}[\indent(1)]
	\item If $w$ (width of the ROBP) is small, then $\Phi_0=0$ and $\Phi_{t+1}-\Phi_t$ is {\it large} for each $0\<=t\<=n-1$.\label{overview.potential.req.1}
	\item If each interval in the last layer has length $\<=n/5$, then $\Phi_n$ is {\it small}.\label{overview.potential.req.2}
\end{enumerate}
We will define the threshold for {\it large} and {\it small} later. Ideally, if (\ref{overview.potential.req.1}) and (\ref{overview.potential.req.2}) are satisfied, then we can argue an ROBP with a small width should have large $(\Phi_{t+1}-\Phi_t)$'s, and thus a large $\Phi_n$, and therefore contradicts (\ref{overview.potential.req.2}).\par
Now let's define $\Phi_t$. Let $\Phi_t:=\sum_{i=0}^t\phi_t(i)$, here $\phi_t(i)$ is the maximum of $b-i$ over all interval labels $[a,b]$ in the $t$-th layer such that $i\in[a,b]$.\par
(\ref{overview.potential.req.1}) holds in the following sense. First, we claim that we have $\phi_{t+1}(i)\>=\phi_t(i)$ for each $i\in\{0,1,\cs,t\}$: suppose $[a,b]$ is the interval in the $t$-th layer such that $i\in[a,b]$ and reaches the maximal $b-i$, then some interval in the $(t+1)$-th layer contains $[a,b]$ and further contains $i,b$, so $\phi_{t+1}(i)\>=b-i$. Second, we also claim that if $i$ is not the left endpoint of some interval in the $t$-th layer, then $\phi_{t+1}(i)\>=\phi_t(i)+1$: suppose $[a,b]$ is an interval in the $t$-th layer such that $i\in[a,b]$ and reaches the maximal $b-i$, then some interval in the $(t+1)$-th layer contains $[a+1,b+1]$ and further contains $i,b+1$ (since $i\neq a$), so $\phi_{t+1}(i)\>=b+1-i$. Therefore we conclude that $\Phi_{t+1}-\Phi_t\>=t+1-w$ and hence $\Phi_n\>=\sum_{t=0}^{n-1}(t+1-w)\>=\f{n^2}2-nw$.\par
(\ref{overview.potential.req.2}) holds in the following sense: if each interval in the last layer has length $\<=n/5$, then each $\phi_n(i)$ is $\<=n/5$, and thus $\Phi_n\<=\f{n(n+1)}5$.\par
Finally combine $\Phi_n\>=\f{n^2}2-nw$ and $\Phi_n\<=\f{n(n+1)}5$ we get $w\>=\Omega(n)$, which gives a lower bound for $\approxcount[n,n/10]$.

\subsection{Lower Bound for \tops{$\approxcountc 3[n,n/20]$}{ApproxCount\_\{3-counter\}[n,n/20]}}

Now we move on to the lower bounds for the general $k$-counter approximate counting. We only sketch the proof of $\approxcountc 3[n,n/20]$ requires $\Omega(n^2)$ width in ROBP, since the proof can naturally generalize to the general $k$-counter setting. We also remark that the theorem ``$\approxcountc 3[n,n/20]$ requires $\Omega(n^2)$ width in ROBP'' cannot be directly derived from the lower bound for $\{0,1\}$-approximate counting.\par
Now instead of labeling each vertex with an interval, we label each vertex $v$ with a rectangle $[a_1,b_1]\times[a_2,b_2]$, here $a_1,b_1$ are the minimal/maximal possible number of $1$'s in the input string in order to reach $v$, and $a_2,b_2$ are the minimal/maximal possible number of $2$'s in the input string in order to reach $v$. Note that we do not consider the error on the frequency of $3$, but it is still sufficient for us to establish the lower bound. Also we need to mention that, by labeling $v$ with $[a_1,b_1]\times[a_2,b_2]$, it is not necessary that some input with simultaneously $b_1$ $1$'s and $b_2$ $2$'s can reach $v$. We only require there exists an input with $b_1$ $1$'s and an input (maybe a different one) with $b_2$ $1$'s that reach $v$.\par
The potential function is similar. For $0\<=t\<=n$, define $\Phi_t:=\sum_{i_1,i_2\colon i_1,i_2\>=0,~i_1+i_2\<=t}\phi_t(i_1,i_2)$, here $\phi_t(i_1,i_2)$ is the maximum of $b_1+b_2-i_1-i_2$ over all rectangle labels $[a_1,b_1]\times[a_2,b_2]$ in the $t$-th layer such that $(i_1,i_2)\in[a_1,b_1]\times[a_2,b_2]$.\par
On one hand, we can still similarly prove that $\phi_{t+1}(i_1,i_2)-\phi_{t}(i_1,i_2)$ is always $\>=0$, and $\>=1$ if $(i_1,i_2)$ is not the left-bottom corner of some rectangle in the $t$-th layer. Therefore $\Phi_{t+1}-\Phi_t\>=\f{(t+1)^2}2-w$ and $\Phi_n\>=\sum_{t=0}^{n-1}\(\f{(t+1)^2}2-w\)\>=\f{n^3}6-nw$. On the other hand, if the ROBP computes $\approxcountc 3[n,n/20]$, then each rectangle in the last layer should have both width and height $\<=n/10$, hence each $\phi_n(i_1,i_2)$ is $\<=n/5$, and $\Phi_n\<=\f{n(n+1)^2}{10}$. Therefore we get $w\>=\Omega(n^2)$.\par
In the formal proof in \Cref{sec.k-counter-lb}, we will perform a more careful analysis on the potential function, and obtain tighter lower bounds.

\section{Preliminaries}\label{sec.prelim}

In this section, we give the formal definitions of our notions.

\subsection{Computation Model: Read-once Branching Programs}

We use the following formal definition of the read-once branching program (ROBP) model.

\begin{definition}{\bf (Read-once Branching Program)}\label{def.ROBP}
	Let $\Sigma$ be a finite alphabet. A length-$n$ read-once branching program (ROBP) $P$ consists of a directed layered multigraph with $n+1$ layers $V_0,V_1,\cs,V_n$. For every $0\<=i\<=n-1$, each vertex $v\in V_i$ has $|\Sigma|$ outgoing edges leading to $V_{i+1}$, each labeled a distinct element in $\Sigma$. Vertices in $V_n$ have zero outgoing edges. $V_0$ consists of a single ``start vertex'': $V_0=\{v_{\mrm{start}}\}$.\par
	An input $x\in\Sigma^n$ selects a path $(v_0,v_1,\cs,v_n)$ through the graph: the path starts at $v_0=v_{\mrm{start}}$, and upon reaching a vertex $v_i\in V_i$, the term $x_{i+1}$ specifies which outgoing edge to use. Each vertex $v$ in $V_n$ is labeled an output string $\mrm{output}(v)$, and the output of $P$ on input $x$ is $P(x)=\mrm{output}(v_n)$.\par
	We further require that each vertex can be reached on some input. i.e., there are no useless vertices.\par
	The width of the program is the maximum number of vertices in a single layer.
\end{definition}

\subsection{Problem: Approximate Counting}

We formally define the problem of approximate counting and its variants.

\begin{definition}{\bf ($\{0,1\}$-approximate Counting)}\label{def.approx-count}
	For any integer $n\>=1$ and real number $\Delta\>=0$, define the problem $\approxcount[n,\Delta]$ as follows: on any input $x\in\{0,1\}^n$ ($\{0,1\}$ is the alphabet and $n$ is the input length), a valid output of $x$ is a real number $\hat{S}$ such that 
	$$\la\hat{S}-\sum_{i=1}^nx_i\ra\<=\Delta.$$\par
	We say a program $P$ computes $\approxcount[n,\Delta]$ if on any input $x\in\{0,1\}^n$, $P(x)$ outputs a valid output of $x$.
\end{definition}

\begin{definition}{\bf ($k$-counter Approximate Counting)}\label{def.approx-count-k-counter}
	For any integers $k\>=2$, $n\>=1$, and real number $\Delta\>=0$, define the problem $\approxcountc k[n,\Delta]$ as follows: on any input $x\in[k]^n$ ($[k]$ is the alphabet and $n$ is the input length), a valid output of $x$ is a $k$-tuple of real numbers $\(\hat{S}_1,\cs,\hat{S}_k\)$ such that 
	$$\la\hat{S_j}-\#\{i\in[n]\colon x_i=j\}\ra\<=\Delta$$
	for all $j\in[k]$.\par
	We say a program $P$ computes $\approxcountc k[n,\Delta]$ if on any input $x\in[k]^n$, $P(x)$ outputs a valid output of $x$.
\end{definition}

\begin{definition}{\bf ($k$-parallel Approximate Counting)}\label{def.approx-count-k-par}
	For any integers $k\>=1$, $n\>=1$, and real number $\Delta\>=0$, define the problem $\approxcountp k[n,\Delta]$ as follows: on any input $x\in\(\{0,1\}^k\)^n$ ($\{0,1\}^k$ is the alphabet and $n$ is the input length), a valid output of $x$ is a $k$-tuple of real numbers $\(\hat{S}_1,\cs,\hat{S}_k\)$ such that 
	$$\la\hat{S_j}-\#\{i\in[n]\colon (x_i)_j=1\}\ra\<=\Delta$$
	for all $j\in[k]$.\par
	We say a program $P$ computes $\approxcountp k[n,\Delta]$ if on any input $x\in\(\{0,1\}^k\)^n$, $P(x)$ outputs a valid output of $x$.
\end{definition}

We remind that $\approxcountc 2$ and $\approxcountp 1$ are equivalent to $\approxcount$.

\section{Lower Bound for \tops{$k$}{k}-counter Approximate Counting}\label{sec.k-counter-lb}

The main purpose of this section is to prove \Cref{thm.k-counter-lb}, which is a very general statement of the lower bound for \tops{$k$}{k}-counter approximate counting. In \Cref{subsec.k-counter-cor} we present the corollaries in various regimes of the parameters, whose statements are more intuitive.


\begin{theorem}{\bf ($k$-counter Lower bound)}\label{thm.k-counter-lb}
	For any integers $k\>=2$, $n\>=1$, $w\>=1$, and real number $0\<=\Delta\<=\f{n}{2(k-1)}$, if there exists a length-$n$ width-$w$ ROBP $P$ over alphabet $[k]$ such that $P$ computes $\approxcountc k[n,\Delta]$, then the following holds: let $m$ be the largest integer such that $m\<=n-1$ and $\C{m+k-1}{k-1}\<=w$, then
	$$\C{m+k}{k}+(n-m-1)w\>=\C{n-2(k-1)\Delta+k-1}{k}.$$
\end{theorem}

	We remark that one can easily verify that the left-hand side, $\C{m+k}{k}+(n-m-1)w$, is non-decreasing in $w$.

\subsection{Rectangle Labeling}

We present a formal definition and some basic properties of the rectangle labeling on the vertices, which is the first step of our proof of the lower bounds.

Let's define some basic notions. For a set $S\subseteq\mathbb{R}^k$ and a vector $\b{v}\in\mathbb{R}^k$, define $S+\b{v}:=\L\b{s}+\b{v}\colon\b{s}\in S\R$. We use $\{\b{e}_1,\cs,\b{e}_k\}$ to denote the standard basis of $\mathbb{R}^k$. i.e., $\b{e}_i$ is the vector with $1$ in the $i$th coordinate and $0$'s in all other coordinates.

\begin{definition}{\bf (Rectangle Labeling of the $k$-counter Case)}\label{def.rect-k-counter}
	For a length-$n$ ROBP $P$ over alphabet $[k]$ that attempts to compute $\approxcountc k[n,\Delta]$, label each vertex $v\in V_t$ $(0\<=t\<=n)$ with a $(k-1)$-dimensional rectangle $R(v)$. $R(v)$ is defined as $[a_1,b_1]\times\cs\times[a_{k-1},b_{k-1}]$, here for each $j\in[k-1]$, $a_j,b_j$ are the minimal and maximal elements in 
	$$\L c\in\mathbb{N}\colon\text{$\exists(x_1,\cs,x_t)\in[k]^t$ such that $(x_1,\cs,x_t)$ reaches $v$ in $P$ and $\#\{i\in[t]\colon x_i=j\}=c$}\R.$$
	i.e., $a_j,b_j$ are the minimal/maximal possible number of $j$'s in the first $t$ input words to reach $v$ in layer $V_t$.\par
	For each $0\<=t\<=n$, let $\mathcal{R}_t:=\{R(v)\colon v\in V_t\}$ be the collection of all rectangle labels in layer $V_t$.
\end{definition}

Rigorously speaking, we should use $R^{(P)}(v)$ and $\mathcal{R}_t^{(P)}$ to show that these rectangles are corresponding to ROBP $P$. Nevertheless, we omit the superscript ``$^{(P)}$'' when $P$ is clear from the context.\par
We also give some immediate observations on the properties of rectangle labeling.

\begin{proposition}{\bf (Properties of Rectangle Labeling of the $k$-counter Case)}\label{prop.rect-k-counter}
	For a length-$n$ ROBP $P$ over alphabet $[k]$ that attempts to compute $\approxcountc k[n,\Delta]$, we have:
	\begin{enumerate}[\indent(1)]
		\item $\mathcal{R}_0=\{[0,0]^{k-1}\}$. For each $0\<=t\<=n-1$, for each $R\in\mathcal{R}_t$, each of $R,R+\b{e}_1,\cs,R+\b{e}_{k-1}$ is contained in some rectangle in $\mathcal{R}_{t+1}$.
		\item Suppose $w$ is the width of $P$, then $|\mathcal{R}_t|\<=w$ for all $0\<=t\<=n$.
		\item \spreadwords{Suppose $P$ computes $\approxcountc k[n,\Delta]$, then for any rectangle}\\$[a_1,b_1]\times\cs\times[a_{k-1},b_{k-1}]\in\mathcal{R}_n$, we have $b_j-a_j\<=2\Delta$ for all $j\in[k-1]$.
	\end{enumerate}
\end{proposition}

\begin{Proof}
	(1): $\mathcal{R}_0=\{[0,0]^{k-1}\}$ is because $V_0$ consists of a single start vertex which has label $[0,0]^{k-1}$ by definition.\par
	For any $R=[a_1,b_1]\times\cs\times[a_{k-1},b_{k-1}]\in\mathcal{R}_t$, suppose $R=R(v)$ for some $v\in V_t$. Suppose the outgoing edges from $v$ with labels $1,\cs,k$ point to $v_1,\cs,v_k\in V_{t+1}$ respectively. For each $j\in[k-1]$, suppose $(p_1^{(j)},\cs,p_t^{(j)}),(q_1^{(j)},\cs,q_t^{(j)})\in[k]^t$ are two possible prefixes of the input that reach $v$ and contain $a_j,b_j$ many $j$'s respectively. Then for each $j\in[k-1]$, $(p_1^{(j)},\cs,p_t^{(j)},k),(q_1^{(j)},\cs,q_t^{(j)},k)$ reach $v_k$ and contain $a_j,b_j$ many $j$'s respectively. So $R\subseteq R(v_k)$. Also note that for each $l\in[k-1]$, for each $j\in[k-1]$, $(p_1^{(j)},\cs,p_t^{(j)},l),(q_1^{(j)},\cs,q_t^{(j)},l)$ reach $v_l$ and contain $a_j+(\b{e}_l)_j,b_j+(\b{e}_l)_j$ many $l$'s respectively. So for each $l\in[k-1]$, we have $R+\b{e}_l\subseteq R(v_l)$.\par
	(2): Observe that $|\mathcal{R}_t|\<=|V_t|\<=w$.\par
	(3): For any $R=[a_1,b_1]\times\cs\times[a_{k-1},b_{k-1}]\in\mathcal{R}_n$, suppose $R=R(v)$ for some $v\in V_n$. Suppose upon reaching $v$, $P$ outputs $(c_1,\cs,c_k)\in\mathbb{R}^k$. For each $j\in[k-1]$, since there exists inputs $\b{x},\b{x}'$ that reach $v$ and contain $a_j,b_j$ many $j$'s respectively, we have $b_j-c_j\<=\Delta$ and $c_j-a_j\<=\Delta$, and thus $b_j-a_j\<=2\Delta$.
\end{Proof}

By using the rectangle labeling, we convert the question of establishing lower bounds for approximate counting into analyzing the rectangle labelings.

\subsection{The Potential Function}

We define and analyze the potential function for proving the lower bound, which is the core step of our proof.

\begin{definition}{\bf (Potential Function for the Proof of \Cref{thm.k-counter-lb})}
	Consider a length-$n$ ROBP $P$ over alphabet $[k]$ that attempts to compute $\approxcountc k[n,\Delta]$.\par
	For each $0\<=t\<=n$, define
	$$T_{t,k}:=\L(x_1,\cs,x_{k-1})\in\mathbb{N}^{k-1}\colon x_1+\cs+x_{k-1}\<=t\R,$$
	which is the set of possible numbers of $1$'s, $2$'s, $\cs$, $(k-1)$'s in the first $t$ input words.\par
	Then, for each $0\<=t\<=n$ and $(x_1,\cs,x_{k-1})\in T_{t,k}$, define\footnote{Here the maximum is taken over all rectangles $R=[a_1,b_1]\times\cs\times[a_{k-1},b_{k-1}]$ such that $R\in\mathcal{R}_{t}$ and $(x_1,\cs,x_{k-1})\in R$.}
	\begin{equation}\label{eq.k-counter-phi}
		\phi_t(x_1,\cs,x_{k-1}):=\max_{\substack{R=[a_1,b_1]\times\cs\times[a_{k-1},b_{k-1}]\colon\\R\in\mathcal{R}_{t},\,(x_1,\cs,x_{k-1})\in R}}\L\min\{b_1+\cs+b_{k-1},t\}-x_1-\cs-x_{k-1}\R.
	\end{equation}
	i.e., $\phi_t(x_1,\cs,x_{k-1})$ is the maximal $l_1$-distance from $(x_1,\cs,x_{k-1})$ to the largest corner of some rectangle in $\mathcal{R}_t$ that contains $(x_1,\cs,x_{k-1})$. Furthermore, if the corner exceeds the border (on which all coordinates sum up to $t$) then we take the $l_1$ distance from $(x_1,\cs,x_{k-1})$ to the border.\par
	Finally, for each $0\<=t\<=n$, define the total potential of layer $V_t$ as
	$$\Phi_t:=\sum_{(x_1,\cs,x_{k-1})\in T_{t,k}}\phi_t(x_1,\cs,x_{k-1}),$$
	which is the sum of the $\phi$'s.
\end{definition}

\Cref{fig.phi} gives an example to illustrate how the potential function is defined.

\begin{figure}[htpb]
	\centering\small
	\begin{tikzpicture}
		\foreach \i in {0,...,8} {
			\foreach \j in {0,...,\numexpr8-\i\relax} {
				\node (\i--\j) at (\i,\j) {$\bullet$};
			}
		}
		\draw[arrow,dotted] (-1,0)--(10,0);
		\draw[arrow,dotted] (0,-1)--(0,10);
		\node at (10.3,0) {$x_1$};
		\node at (0,10.3) {$x_2$};
		\node at (-0.5,-0.3) {$(0,0)$};
		\node at (-0.5,8.3) {$(0,8)$};
		\node at (8.5,-0.3) {$(8,0)$};
		\newcommand{\rect}[5]{\draw[thick,#1] (#2-0.1,#4-0.1)--(#3+0.1,#4-0.1)--(#3+0.1,#5+0.1)--(#2-0.1,#5+0.1)--(#2-0.1,#4-0.1)}
		\rect{green}0203;
		\node at (-1.8,1.5) {$R_1=[0,2]\times[0,3]$};
		\rect{gray}0248;
		\node at (-1.8,6) {$R_2=[0,2]\times[4,8]$};
		\rect{brown}2501;
		\node at (3.5,-0.5) {$R_3=[2,5]\times[0,1]$};
		\rect{red}1525;
		\node at (4,5.5) {$R_4=[1,5]\times[2,5]$};
		\rect{blue}4804;
		\node at (9.8,2) {$R_5=[4,8]\times[0,4]$};
		\node(Ex1) at (1.5,-1.5) {$\phi_t(2,0)=4$};
		\draw[rounded corners] (Ex1)--(1.5,-0.5)--(2--0);
		\node(Ex2) at (-1.8,3.5) {$\phi_t(2,3)=3$};
		\draw[rounded corners] (Ex2)--(1.5,3.5)--(2--3);
	\end{tikzpicture}
	\caption{\small An example for the definition of the potential function. In this example we take $k=3$ and $t=8$. Let $\mathcal{R}_t=\{R_1,R_2,R_3,R_4,R_5\}$ and the values of $R_1,\cs,R_5$ are shown in the figure. We calculate $\phi_t(2,0)$ and $\phi_t(2,3)$ as an example: for $\phi_t(2,0)$, the max operator in \Cref{eq.k-counter-phi} is taken over $R\in\{R_1,R_3\}$, and the maximum is $4$, which is achieved when $R=R_3$; for $\phi_t(2,3)$, the max operator in \Cref{eq.k-counter-phi} is taken over $R\in\{R_1,R_4\}$, and the maximum is $3$, which is achieved when $R=R_4$.}
	\label{fig.phi}
\end{figure}
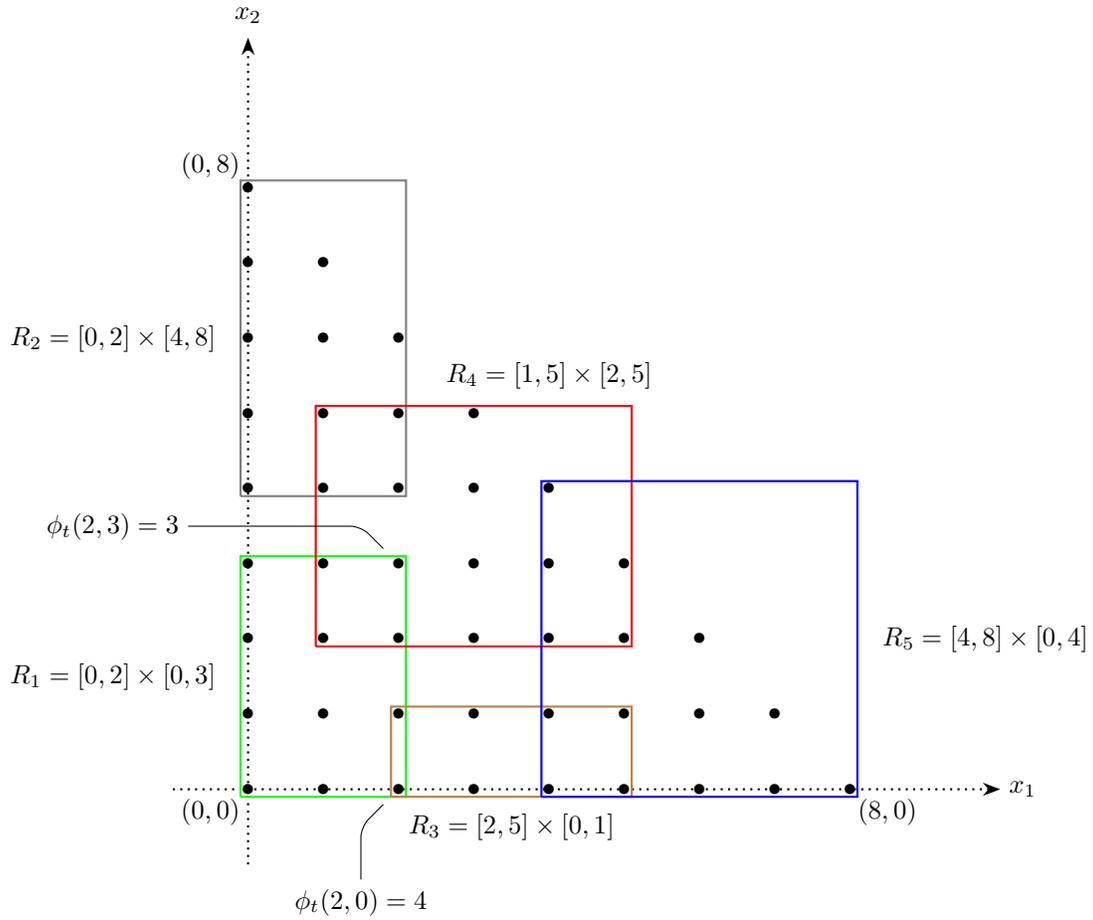

\begin{lemma}{\bf (Growth of the Potential)}\label{lem.k-counter.growth}
	Consider a length-$n$ ROBP $P$ over alphabet $[k]$ that attempts to compute $\approxcountc k[n,\Delta]$. Let $w$ be the width of $P$, then $\Phi_0=0$, and for any $0\<=t\<=n-1$ we have
	$$\Phi_{t+1}-\Phi_t\>=\max\L0,\C{t+k-1}{k-1}-w\R.$$
\end{lemma}
\begin{Proof}
	$\Phi_0=0$ is because by \Cref{prop.rect-k-counter} we have $\mathcal{R}_0=\{[0,0]^{k-1}\}$. Thus $\Phi_0=0$.\par
	Note that $\phi_{t+1}(x_1,\cs,x_{k-1})=0$ if $x_1+\cs+x_{k-1}=t+1$. Thus
	$$\Phi_{t+1}-\Phi_t=\sum_{(x_1,\cs,x_{k-1})\in T_{t,k}}\(\phi_{t+1}(x_1,\cs,x_{k-1})-\phi_t(x_1,\cs,x_{k-1})\).$$\par
	Consider any $(x_1,\cs,x_{k-1})\in T_{t,k}$. Suppose when computing $\phi_t(x_1,\cs,x_{k-1})$, the maximum of $\(\min\{b_1+\cs+b_{k-1},t\}-x_1-\cs-x_{k-1}\)$ is achieved when $R=R_0=[a_1,b_1]\times\cs\times[a_{k-1},b_{k-1}]$, then by \Cref{prop.rect-k-counter} there exists $R_1\in\mathcal{R}_{t+1}$ such that $R_0\subseteq R_1$. Thus when computing $\phi_{t+1}(x_1,\cs,x_{k-1})$, by taking $R=R_1$ we can achieve a larger (or equal) value of $\(\min\{b_1+\cs+b_{k-1},t\}-x_1-\cs-x_{k-1}\)$. So $\phi_{t+1}(x_1,\cs,x_{k-1})\>=\phi_t(x_1,\cs,x_{k-1})$.\par
	Furthermore, if $(x_1,\cs,x_{k-1})\neq(a_1,\cs,a_{k-1})$, pick $j\in[k-1]$ such that $x_j\>=a_j+1$, then $(x_1,\cs,x_{k-1})\in R_0+\b{e}_j$. By \Cref{prop.rect-k-counter} there exists $R_2=[a_1',b_1']\times\cs\times[a_{k-1}',b_{k-1}']\in\mathcal{R}_{t+1}$ such that $(x_1,\cs,x_{k-1})\in R_0+\b{e}_j\subseteq R_2$. Therefore $b_1'\>=b_1,\cs,b_{k-1}'\>=b_{k-1}$ and $b_j'\>=b_j+1$ and 
	$$\min\{b_1'+\cs+b_{k-1}',t+1\}-x_1-\cs-x_{k-1}\>=(\min\{b_1+\cs+b_{k-1},t\}-x_1-\cs-x_{k-1})+1.$$
	So for any $(x_1,\cs,x_{k-1})$, if there does not exist $[a_1,b_1]\times\cs\times[a_{k-1},b_{k-1}]\in\mathcal{R}_t$ such that $(x_1,\cs,x_{k-1})=(a_1,\cs,a_{k-1})$, then $\phi_{t+1}(x_1,\cs,x_{k-1})\>=\phi_t(x_1,\cs,x_{k-1})+1$.\par
	To summarize, over all tuples $(x_1,\cs,x_{k-1})\in T_{t,k}$,  $\phi_{t+1}(x_1,\cs,x_{k-1})-\phi_t(x_1,\cs,x_{k-1})$ is always $\>=0$, and is $\>=1$ for all but at most $|\mathcal{R}_t|$ tuples. Note that $|T_{t,k}|=\C{t+k-1}{k-1}$, so
	$$\Phi_{t+1}-\Phi_t\>=\max\L0,\C{t+k-1}{k-1}-|\mathcal{R}_t|\R\>=\max\L0,\C{t+k-1}{k-1}-w\R.$$	
\end{Proof}

\begin{lemma}{\bf (Implication of Large $\Phi_n$)}\label{lem.k-counter.final}
	Let $P$ be a length-$n$ ROBP over alphabet $[k]$ that computes $\approxcountc k[n,\Delta]$, (we require that $\Delta\<=\f n{2(k-1)}$,) then
	$$\Phi_n\<=\C{n+k-1}{k}-\C{n-2(k-1)\Delta+k-1}{k}.$$
\end{lemma}
\begin{Proof}
	Note that for any $R=[a_1,b_1]\times\cs\times[a_{k-1},b_{k-1}]\in\mathcal{R}_n$ and $(x_1,\cs,x_{k-1})\in R\cap T_{n,k}$, we have
	\begin{align*}
		&\phantomeq\min\{b_1+\cs+b_{k-1},n\}-x_1-\cs-x_{k-1}\\
		&\<=\min\{b_1+\cs+b_{k-1}-a_1-\cs-a_{k-1},n-x_1-\cs-x_{k-1}\}\\
		&\<=\min\{(k-1)\lf2\Delta\rf,n-x_1-\cs-x_{k-1}\}.
	\end{align*}
	Here the last step is because by \Cref{prop.rect-k-counter}, for any $j\in[k-1]$ we have $b_j-a_j\<=2\Delta$, and hence $b_j-a_j\<=\lf2\Delta\rf$. So 
	$$\phi_n(x_1,\cs,x_{k-1})\<=\min\{(k-1)\lf2\Delta\rf,n-x_1-\cs-x_{k-1}\}$$
	for any $(x_1,\cs,x_{k-1})\in T_{n,k}$. Therefore
	\begin{align*}
		\Phi_n&\<=\sum_{\substack{(x_1,\cs,x_{k-1})\in\mathbb{N}^{k-1}\\x_1+\cs+x_{k-1}\<=n}}\min\{(k-1)\lf2\Delta\rf,n-x_1-\cs-x_{k-1}\}\\
		&=\C{n+k-1}{k}-\C{n-(k-1)\lf2\Delta\rf+k-1}{k}\\
		&\<=\C{n+k-1}{k}-\C{n-2(k-1)\Delta+k-1}{k}.
	\end{align*}
\end{Proof}

\subsection{Proof of Theorem \ref{thm.k-counter-lb}}

\begin{proof}[Proof of \Cref{thm.k-counter-lb}]~\par
	By \Cref{lem.k-counter.growth} we have 
	\begin{align*}
		\Phi_n&\>=\sum_{t=0}^{n-1}\max\L0,\C{t+k-1}{k-1}-w\R\\
		&=\sum_{t=m+1}^{n-1}\(\C{t+k-1}{k-1}-w\)\\
		&=\C{n+k-1}k-\C{m+k}k-(n-m-1)w.
	\end{align*}\par
	Combine with \Cref{lem.k-counter.final} we get
	$$\C{n+k-1}{k}-\C{n-2(k-1)\Delta+k-1}{k}\>=\C{n+k-1}k-\C{m+k}k-(n-m-1)w.$$
	Thus
	$$\C{m+k}{k}+(n-m-1)w\>=\C{n-2(k-1)\Delta+k-1}{k}.$$
\end{proof}

\subsection{Corollaries in Various Regimes}\label{subsec.k-counter-cor}

We give some corollaries of \Cref{thm.k-counter-lb} in various regimes of the parameters with more intuitive statements. \Cref{cor.k-counter-lb.small-w} gives a bound on the smallest additive error we can achieve when the width of the ROBP is small, and \Cref{cor.k-counter-lb.small-err} gives a bound on the width of the ROBP to achieve a certain small error.

\begin{corollary}\label{cor.k-counter-lb.standard}
	For any integers $k\>=2$, $n\>=1$, $w\>=1$, if there exists a length-$n$ width-$w$ ROBP $P$ over alphabet $[k]$ such that $P$ computes $\approxcountc k\[n,\f{n}{3(k-1)}\]$, then $w\>=\Omega(n/k)^{k-1}$.
\end{corollary}

\begin{corollary}\label{cor.k-counter-lb.small-w}
	For any integers $k\>=2$, $n\>=1$, $w\>=1$, and real number $0\<=\Delta\<=\f{n}{2(k-1)}$, if there exists a length-$n$ width-$w$ ROBP $P$ over alphabet $[k]$ such that $P$ computes $\approxcountc k[n,\Delta]$, then 
	$$\Delta\>=\f n{2(k-1)}-\f{\sqrt[k]{k!}}{2(k-1)}\c\sqrt[k]{nw}.$$
\end{corollary}

\begin{corollary}\label{cor.k-counter-lb.small-err}
	For any integers $k\>=2$, $n\>=10k$, $w\>=1$, and real number $1\<=\Delta\<=\f{n}{10(k-1)}$, if~there\\\spreadwords{exists a length-$n$ width-$w$ ROBP $P$ over alphabet $[k]$ such that $P$ computes $\approxcountc k[n,\Delta]$,}\\then
	$$w\>=\C{n+k-1-2(k-1)\Delta-O\(\sqrt{n\Delta}\)}{k-1}.$$
\end{corollary}

\begin{proof}[Proof of \Cref{cor.k-counter-lb.small-w}]~\par
	Let $m$ be as stated in \Cref{thm.k-counter-lb}.\par
	By \Cref{thm.k-counter-lb} we have
	\begin{align*}
		\C{n-2(k-1)\Delta+k-1}{k}&\<=\C{m+k}{k}+(n-m-1)w\\
		&\<=\(\f{\C{m+k}{k}}{\C{m+k-1}{k-1}}+n-m-1\)w\\
		&=\(\f{m+k}k+n-m-1\)w\\
		&\<=nw.
	\end{align*}
	Also note that $\C{n-2(k-1)\Delta+k-1}{k}\>=(n-2(k-1)\Delta)^k/(k!)$, so $nw\>=(n-2(k-1)\Delta)^k/(k!)$ and hence
	$$\Delta\>=\f n{2(k-1)}-\f{\sqrt[k]{k!}}{2(k-1)}\c\sqrt[k]{nw}.$$\par
\end{proof}

\begin{proof}[Proof of \Cref{cor.k-counter-lb.standard}]~\par
	Take $\Delta:=\f{n}{3(k-1)}$ in \Cref{cor.k-counter-lb.small-w} we get $\f n{3(k-1)}\>=\f n{2(k-1)}-\f{\sqrt[k]{k!}}{2(k-1)}\c\sqrt[k]{nw}$, and therefore $w\>=\f{n^{k-1}}{k!\c 3^k}\>=\Omega(n/k)^{k-1}$.	
\end{proof}

\begin{proof}[Proof of \Cref{cor.k-counter-lb.small-err}]~\par
	Let $m$ be as stated in \Cref{thm.k-counter-lb}. Since $w\>=\C{m+k-1}{k-1}$, we only need to prove that $m+k-1\>=n+k-1-2(k-1)\Delta-O\(\sqrt{n\Delta}\)$, or equivalently, $n-m-2(k-1)\Delta\<=O\(\sqrt{n\Delta}\)$. If $m\>=n-2(k-1)\Delta-10$ then we are done. Below we only consider the case $m<n-2(k-1)\Delta-10$.\par
	Since $\C{m+k-1}{k-1}\<=w<\C{m+k}{k-1}$, \Cref{thm.k-counter-lb} implies 
	$$\C{m+k}{k}+(n-m-1)\C{m+k}{k-1}\>=\C{n-2(k-1)\Delta+k-1}{k},$$
	divide both sides by $\C{m+k}{k}$, we get
	\begin{align*}
		1+\f{k(n-m-1)}{m+1}&\>=\f{\C{n-2(k-1)\Delta+k-1}{k}}{\C{m+k}{k}}
		\>=\(\f{n-2(k-1)\Delta+k-1}{m+k}\)^k\\
		&\>=1+k\c\f{n-m-2(k-1)\Delta-1}{m+k}+\f{k(k-1)}2\c\(\f{n-m-2(k-1)\Delta-1}{m+k}\)^2.
	\end{align*}
	Thus
	\begin{align*}
		\(n-m-2(k-1)\Delta-1\)^2&\<=(m+k)^2\c\f2{k-1}\c\(\f{n-m-1}{m+1}-\f{n-m-2(k-1)\Delta-1}{m+k}\)\\
		&=\f2{k-1}\c\(\f{(k-1)(n-m+1)(m+k)}{m+1}+2(m+k)(k-1)\Delta\)\\
		&\<=\f{4(n+1)(m+k)}{m+1}+4n\Delta.
	\end{align*}\par
	If $m\>=k$, then $\(n-m-2(k-1)\Delta-1\)^2\<=4\(\f{(n+1)(m+k)}{m+1}+n\Delta\)\<=O(n\Delta)$, and thus we have $n-m-2(k-1)\Delta\<=O\(\sqrt{n\Delta}\)$.\par
	So finally we only need to prove that $m\>=k$. By \Cref{cor.k-counter-lb.small-w} we have $\f{n}{10(k-1)}\>=\Delta\>=\f n{2(k-1)}-\f{\sqrt[k]{k!}}{2(k-1)}\c\sqrt[k]{nw}$ and therefore $w\>=\f{4^k\c n^{k-1}}{5^k\c k!}\>=\f{4^k\c(10(k-1))^{k-1}}{5^k\c k!}\>=\f{8^k}{20}\>=\C{2k-1}{k-1}$, so $m\>=k$.
\end{proof}

\section{Lower Bound for \tops{$k$}{k}-parallel Approximate Counting}\label{sec.k-par-lb}

In this section, we prove the lower bound for the $k$-parallel case. 

\begin{theorem}\label{thm.k-par-lb}
	For any integers $k\>=1$, $n\>=3k$, $w\>=1$, if there exists a length-$n$ width-$w$ ROBP $P$ over alphabet $\{0,1\}^k$ such that $P$ computes $\approxcountp k[n,n/3]$, then $w\>=n^{\Omega(k)}$.
\end{theorem}

We should mention that by directly applying the proof in \Cref{sec.k-counter-lb}, we can only show the hardness of $\approxcountp k[n,n/(3k)]$. In the proof of \Cref{thm.k-counter-lb}, we only show that most $\phi$'s grow by at least $1$ (when moving to the next layer), and thus finally we can only show there exists some point in the last layer with the value of $\phi$ larger than $\Omega(n)$. However, to prove \Cref{thm.k-par-lb}, we need to show there exists some point in the last layer with the value of $\phi$ larger than $\Omega(nk)$. To fix this issue, we need to argue that most $\phi$'s grow by at least $\Omega(k)$, and we need to exclude not only the points at the corners of the rectangles, but also the points at the low-dimensional boundaries of the rectangles. (i.e., points in the rectangle such that many coordinates have reached the boundary.)\par
Our proof of \Cref{thm.k-par-lb} is similar to the proof of \Cref{thm.k-counter-lb}.

\subsection{Rectangle Labeling}

\begin{definition}{\bf (Rectangle Labeling of the $k$-parallel Case)}\label{def.rect-k-par}
	For a length-$n$ ROBP $P$ over alphabet $\{0,1\}^k$ that attempts to compute $\approxcountp k[n,\Delta]$, label each vertex $v\in V_t$ $(0\<=t\<=n)$ with a $k$-dimensional rectangle $R(v)$. $R(v)$ is defined as $[a_1,b_1]\times\cs\times[a_k,b_k]$, here for each $j\in[k]$, $a_j,b_j$ are the minimal and maximal elements in 
	$$\L c\in\mathbb{N}\colon\text{$\exists(x_1,\cs,x_t)\in\(\{0,1\}^k\)^t$, $(x_1,\cs,x_t)$ reaches $v$ and $\#\{i\in[t]\colon (x_i)_j=1\}=c$}\R.$$
	i.e., $a_j,b_j$ are the minimal/maximal possible number of $1$'s in the $j$th coordinate of the first $t$ input words to reach $v$ in layer $V_t$.\par
	For each $0\<=t\<=n$, let $\mathcal{R}_t:=\{R(v)\colon v\in V_t\}$ be the collection of all rectangle labels in layer $V_t$.
\end{definition}

\begin{proposition}{\bf (Properties of Rectangle Labeling of the $k$-parallel Case)}\label{prop.rect-k-par}
	For a length-$n$ ROBP $P$ over alphabet $\{0,1\}^k$ that attempts to compute $\approxcountp k[n,\Delta]$, we have:
	\begin{enumerate}[\indent(1)]
		\item $\mathcal{R}_0=\{[0,0]^k\}$. For each $0\<=t\<=n-1$, for each $R\in\mathcal{R}_t$, each of $R+\b{z}$ ($\b{z}\in\{0,1\}^k$) is contained in some rectangle in $\mathcal{R}_{t+1}$.
		\item Suppose $w$ is the width of $P$, then $|\mathcal{R}_t|\<=w$ for all $0\<=t\<=n$.
		\item Suppose $P$ computes $\approxcountp k[n,\Delta]$, then for any rectangle $[a_1,b_1]\times\cs\times[a_k,b_k]\in\mathcal{R}_n$, we have $b_j-a_j\<=2\Delta$ for all $j\in[k]$.
	\end{enumerate}
\end{proposition}

\begin{proof}[Proof of \Cref{prop.rect-k-par}]~\par
	(1): $\mathcal{R}_0=\{[0,0]^k\}$ is because $V_0$ consists of a single start vertex which has label $[0,0]^k$ by definition. \par
	For any $R=[a_1,b_1]\times\cs\times[a_k,b_k]\in\mathcal{R}_t$, suppose $R=R(v)$ for some $v\in V_t$. Suppose the outgoing edges from $v$ with label $\b{z}\in\{0,1\}^k$ points to $v_{\b{z}}\in V_{t+1}$. For each $j\in[k]$, suppose $(p_1^{(j)},\cs,p_t^{(j)}),(q_1^{(j)},\cs,q_t^{(j)})\in\(\{0,1\}^k\)^t$ are two possible prefixes of the input that reaches $v$ and contains $a_j,b_j$ many $1$'s in the $j$th coordinate respectively. Then for each $\b{z}\in\{0,1\}^k$, for each $j\in[k]$, $(p_1^{(j)},\cs,p_t^{(j)},\b{z}),(q_1^{(j)},\cs,q_t^{(j)},\b{z})$ reaches $v_{\b{z}}$ and contains $a_j+\b{z}_j,b_j+\b{z}_j$ many $1$'s in the $j$th coordinate respectively. So for each $\b{z}\in\{0,1\}^k$, we have $R+\b{z}\subseteq R(v_{\b{z}})$. \par
	(2): Observe that $|\mathcal{R}_t|\<=|V_t|\<=w$.\par
	(3): For any $R=[a_1,b_1]\times\cs\times[a_k,b_k]\in\mathcal{R}_n$, suppose $R=R(v)$ for some $v\in V_n$. Suppose upon reaching $v$, $P$ outputs $(c_1,\cs,c_k)\in\mathbb{R}^k$. For each $j\in[k]$, since there exists inputs $\b{x},\b{x}'$ that reaches $v$ and contains $a_j,b_j$ many $1$'s in the $j$th coordinate respectively, we have $b_j-c_j\<=\Delta$ and $c_j-a_j\<=\Delta$, and thus $b_j-a_j\<=2\Delta$.
\end{proof}

\subsection{Proof of Theorem \ref{thm.k-par-lb}}

\begin{definition}
	Consider a length-$n$ ROBP $P$ over alphabet $\{0,1\}^k$ that attempts to compute $\approxcountp k[n,n/3]$.\par
	For each $\lf n/10\rf\<=t\<=n$ and $(x_1,\cs,x_k)\in\{0,1,\cs,\lf n/10\rf\}^k$, define
	$$\phi_t(x_1,\cs,x_k):=\max_{\substack{R:=[a_1,b_1]\times\cs\times[a_k,b_k]\colon\\R\in\mathcal{R}_{t},\,(x_1,\cs,x_k)\in R}}\L b_1+\cs+b_k-x_1-\cs-x_k\R.$$
	i.e., $\phi_t(x_1,\cs,x_k)$ is the maximal $l_1$-distance from $(x_1,\cs,x_k)$ to the largest corner of some rectangle in $\mathcal{R}_t$ that contains $(x_1,\cs,x_k)$.\par
	Then for each $\lf n/10\rf\<=t\<=n$, define the total potential of layer $V_t$ by
	$$\Phi_t:=\sum_{(x_1,\cs,x_k)\in\{0,1,\cs,\lf n/10\rf\}^k}\phi_t(x_1,\cs,x_k),$$
	which is the sum of the $\phi$'s.
\end{definition}

\begin{lemma}\label{lem.k-par.growth}
	\spreadwords{Consider a length-$n$ ROBP $P$ over alphabet $\{0,1\}^k$ that attempts to compute} \\$\approxcountp k[n,n/3]$. Let $w$ be the width of $P$, then $\Phi_{\lf n/10\rf}\>=0$, and for any $\lf n/10\rf\<=t\<=n-1$ we have
	$$\Phi_{t+1}-\Phi_t\>=\lF\f{9k}{10}\rF\c\(\(\lf n/10\rf+1\)^k-w\c2^k\c\(\lf n/10\rf+1\)^{\lF\f{9k}{10}\rF-1}\).$$
\end{lemma}
\begin{Proof}
	$\Phi_{\lf n/10\rf}\>=0$ is because by definition, the value of $\phi$ is always non-negative, and therefore the value of $\Phi$ is always non-negative.\par
	Note that we have
	$$\Phi_{t+1}-\Phi_t=\sum_{(x_1,\cs,x_k)\in\{0,1,\cs,\lf n/10\rf\}^k}\(\phi_{t+1}(x_1,\cs,x_k)-\phi_t(x_1,\cs,x_k)\).$$\par
	Consider any $(x_1,\cs,x_k)\in\{0,1,\cs,\lf n/10\rf\}^k$. Suppose when computing $\phi_t(x_1,\cs,x_k)$, the maximum of $\(b_1+\cs+b_k-x_1-\cs-x_k\)$ is achieved when $R=R_0=[a_1,b_1]\times\cs\times[a_k,b_k]$, then by \Cref{prop.rect-k-par} there exists $R_1\in\mathcal{R}_{t+1}$ such that $R_0\subseteq R_1$. Thus when computing \spreadwords{$\phi_{t+1}(x_1,\cs,x_k)$, by taking $R=R_1$ we can achieve a larger (or equal) value of }\\$\(b_1+\cs+b_k-x_1-\cs-x_k\)$. So $\phi_{t+1}(x_1,\cs,x_k)\>=\phi_t(x_1,\cs,x_k)$.\par
	Furthermore, if $(x_1,\cs,x_k)$ and $(a_1,\cs,a_k)$ differ on at least $\lF\f{9k}{10}\rF$ coordinates, pick $\b{z}\in\{0,1\}^k$ such that
	$$\b{z}_j=\begin{cases}0&\text{if }a_j=x_j\\1&\text{if }a_j<x_j\end{cases}~(j\in[k]).$$
	Then at least $\lF\f{9k}{10}\rF$ coordinates of $\b{z}$ are $1$, and $(x_1,\cs,x_k)\in R_0+\b{z}$. By \Cref{prop.rect-k-counter} there exists $R_2=[a_1',b_1']\times\cs\times[a_k',b_k']\in\mathcal{R}_{t+1}$ such that $(x_1,\cs,x_k)\in R_0+\b{z}\subseteq R_2$. Therefore
	\begin{align*}
		b_1'+\cs+b_k'-x_1-\cs-x_k&\>=\sum_{j=1}^k(b_j+\b{z}_j)-x_1-\cs-x_k\\
		&\>=(b_1+\cs+b_k-x_1-\cs-x_k)+\lF\f{9k}{10}\rF.
	\end{align*}
	So for any $(x_1,\cs,x_k)\in\{0,1,\cs,\lf n/10\rf\}^k$, if there does not exist $[a_1,b_1]\times\cs\times[a_k,b_k]\in\mathcal{R}_t$ such that $(x_1,\cs,x_k)$ and $(a_1,\cs,a_k)$ agree on $\>=\lf k/10\rf+1$ coordinates, then $\phi_{t+1}(x_1,\cs,x_k)\>=\phi_t(x_1,\cs,x_k)+\lF\f{9k}{10}\rF$.\par
	Let's upper bound the number of $(x_1,\cs,x_k)$'s that agree on $\>=\lf k/10\rf+1$ coordinates with $(a_1,\cs,a_k)$ for some $[a_1,b_1]\times\cs\times[a_k,b_k]\in\mathcal{R}_t$. There are $|\mathcal{R}_t|$ choices for $(a_1,\cs,a_k)$, then there are $\<=2^k$ choices for which coordinates they differ, then there are $\<=\lf n/10\rf+1$ choices for each $x_i$ which is different from $a_i$. So the total number of such $(x_1,\cs,x_k)$'s is not more than
	$$|\mathcal{R}_t|\c2^k\c(\lf n/10\rf+1)^{\lF\f{9k}{10}\rF-1}.$$ \par
	Thus over all tuples $(x_1,\cs,x_k)\in\{0,1,\cs,\lf n/10\rf\}^k$, $\phi_{t+1}(x_1,\cs,x_k)-\phi_t(x_1,\cs,x_k)$ is always $\>=0$, and is $\>=\lF\f{9k}{10}\rF$ for all but at most $|\mathcal{R}_t|\c2^k\c(\lf n/10\rf+1)^{\lF\f{9k}{10}\rF-1}$ tuples. So
	\begin{align*}
		\Phi_{t+1}-\Phi_t&\>=\lF\f{9k}{10}\rF\c\(\(\lf n/10\rf+1\)^k-|\mathcal{R}_t|\c2^k\c\(\lf n/10\rf+1\)^{\lF\f{9k}{10}\rF-1}\)\\
		&\>=\lF\f{9k}{10}\rF\c\(\(\lf n/10\rf+1\)^k-w\c2^k\c\(\lf n/10\rf+1\)^{\lF\f{9k}{10}\rF-1}\).
	\end{align*}
\end{Proof}

\begin{lemma}\label{lem.k-par.final}
	\spreadwords{Let $P$ be a length-$n$ ROBP over alphabet $\{0,1\}^k$ that computes }\\$\approxcountp k[n,n/3]$, then
	$$\Phi_n\<=(\lf n/10\rf+1)^k\c\f{2kn}3.$$
\end{lemma}
\begin{Proof}
	Note that for any $R=[a_1,b_1]\times\cs\times[a_k,b_k]\in\mathcal{R}_n$ and $(x_1,\cs,x_k)\in R\cap\{0,1,\cs,\lf n/10\rf\}^k$, we have
	$$b_1+\cs+b_k-x_1-\cs-x_k\<=\f{2kn}{3},$$
	since by \Cref{prop.rect-k-par}, for any $j\in[k]$ we have $b_j-a_j\<=2n/3$. So
	$$\phi_n(x_1,\cs,x_k)\<=\f{2kn}3$$
	for any $(x_1,\cs,x_k)\in\{0,1,\cs,\lf n/10\rf\}^k$. Therefore
	$$\Phi_n=\sum_{(x_1,\cs,x_k)\in\{0,1,\cs,\lf n/10\rf\}^k}\phi_n(x_1,\cs,x_k)\<=(\lf n/10\rf+1)^k\c\f{2kn}{3}.$$
\end{Proof}

\begin{proof}[Proof of \Cref{thm.k-par-lb}]~\par
	By \Cref{lem.k-par.growth} we have 
	\begin{align*}
		\Phi_n&\>=\lF\f{9n}{10}\rF\c\lF\f{9k}{10}\rF\c\(\(\lf n/10\rf+1\)^k-w\c2^k\c\(\lf n/10\rf+1\)^{\lF\f{9k}{10}\rF-1}\)\\
		&\>=\f{81kn}{100}\c\(\(\lf n/10\rf+1\)^k-w\c2^k\c\(\lf n/10\rf+1\)^{\lF\f{9k}{10}\rF-1}\)
	\end{align*}\par
	Combine with \Cref{lem.k-par.final} we get
	$$(\lf n/10\rf+1)^k\c\f{2kn}{3}\>=\f{81kn}{100}\c\(\(\lf n/10\rf+1\)^k-w\c2^k\c\(\lf n/10\rf+1\)^{\lF\f{9k}{10}\rF-1}\).$$
	Thus 
	$$w\>=\f{43}{243}\c\f{\(\lf n/10\rf+1\)^{\lf\f{k}{10}\rf+1}}{2^k}\>=\Omega(n)^{\Omega(k)}.$$
	So $w\>=n^{\Omega(k)}$.\footnote{Note that to compute $\approxcountp k[n,n/3]$, an ROBP needs $\>=2^k$ different outputs, so $w\>=2^k$. If $w\>=(n/C_1)^{k/C_2}$ for constants $C_1,C_2>0$, then $w\>=\((n/C_1)^{k/C_2}\)^{C_2/(C_1+C_2)}\c(2^k)^{C_1/(C_1+C_2)}\>=\(n\c 2^{C_1}/C_1\)^{k/(C_1+C_2)}\>=n^{k/(C_1+C_2)}$.}
\end{proof}

\spreadwords{Note that our lower bound $w\>=n^{\Omega(k)}$ is already tight in the sense of $\Theta(\log w)$, or}\\$\Theta(\text{number of bits})$ in the space usage. However we still want to ask if we can get a tighter bound in the sense of $\Theta(w)$ for a fixed constant $k$. Thus we propose the following open problem:

\begin{openproblem}\label{open.k-par}
	Prove or disprove: for $n\gg k\>=1$, computing $\approxcountp k[n,n/3]$ requires $\Omega(n)^k$ width.
\end{openproblem}

\Cref{open.k-par} cannot be proved directly via our proof of \Cref{thm.k-par-lb}.

\section{Non-trivial Algorithms for Approximate Counting}\label{sec.algo}

In this section, we present some non-trivial algorithms for approximate counting. Surprisingly, they almost match our lower bounds in some regimes.

\subsection{Small Width Regime of \tops{$\{0,1\}$}{\{0,1\}}-Approximate Counting}

\begin{theorem}\label{thm.algo.small-w}
	For any integers $n\>=1$ and $3\<=w\<=n/10$, there exists a length-$n$ width-$w$ ROBP $P$ over alphabet $\{0,1\}$ such that $P$ computes $\approxcount\[n,\Delta\]$ for some $\Delta\<=n/2-\Omega(\sqrt{nw})$.
\end{theorem}
\begin{Proof}
	Divide $[1,n]$ into $l:=\lf\sqrt{n/w}\rf$ intervals $[p_0+1,p_1],\cs,[p_{l-1}+1,p_l]$, here $0=p_0<p_1<\cs<p_l=n$, and $p_j-p_{j-1}\in\{\lf n/l\rf,\lF n/l\rF\}$ for all $1\<=j\<=l$. Consider the formula
	$$f(x_1,\cs,x_n):=\bigwedge_{j=1}^l\(\#\L i\in[p_{j-1}+1,p_j]\colon x_i=1\R\>=w-2\).$$\par
	First we note that there exists a length-$n$ width-$w$ ROBP $P$ that computes $f$, i.e., on input $(x_1,\cs,x_n)\in\{0,1\}^n$, $P$ outputs $f(x_1,\cs,x_n)$. Let $V_t:=\L u^{(t)},v_0^{(t)},\cs,v_{w-2}^{(t)}\R$ for $1\<=t\<=n$. For $t$ such that $t\in[p_{j_0-1}+1,p_{j_0}]$, we wish that if an input $(x_1,\cs,x_n)$ reaches $u^{(t)}$, then some clause $(\#\L i\in[p_{j-1}+1,p_j]\colon x_i=1\R\>=w-2)$ (for some $j<j_0$) is false; otherwise all these previous clauses are true, and if $(x_1,\cs,x_n)$ reaches $v_0^{(t)},\cs,v_{w-2}^{(t)}$ then $(x_{p_{j_0-1}+1},x_{p_{j_0-1}+2},\cs,x_t)$ contains $=0,=1,\cs,=w-3,\>=w-2$ many $1$'s, respectively. This can be implemented since which vertex should $(x_1,\cs,x_n)$ reach in the $t$-th layer can be uniquely determined by $x_t$ and which vertex it reaches in the $(t-1)$-th layer, and finally let $P$ output $1$ iff $(x_1,\cs,x_n)$ reaches $v_{w-2}^{(n)}$.\par
	On the other hand, if $f(x_1,\cs,x_n)=1$, then each segment $x_{p_{j-1}+1},x_{p_{j-1}+2},\cs,x_{p_j}$ contains $\>=w-2$ many $1$'s, thus $x_1,\cs,x_n$ contains $\>=\lf\sqrt{n/w}\rf\c(w-2)\>=\sqrt{nw}/10$ many $1$'s. If $f(x_1,\cs,x_n)=0$, then some segment $x_{p_{j-1}+1},x_{p_{j-1}+2},\cs,x_{p_j}$ contains $\<=w-3$ many $1$'s, and therefore contains $\>=\lf n/l\rf-w+3=\lf n/\lf\sqrt{n/w}\rf\rf-w+3\>=\sqrt{nw}/10$ many $0$'s. So the ROBP $P$ that computes $f$ can also compute $\approxcount\[n,n/2-\sqrt{nw}/20\]$, by outputting $(n/2-\sqrt{nw}/20)$ or $(n/2+\sqrt{nw}/20)$ if the output of $f$ is $0$ or $1$, respectively.
\end{Proof}

We remark that in the case $w=3$, the $f$ we compute is just the $\msf{Tribe}$ function.\par
Combine with \Cref{cor.k-counter-lb.small-w} we get:

\begin{corollary}\label{cor.tight-small-w}
	For any integers $n\>=1$ and $3\<=w\<=n/10$, let $\Delta(n,w)$ be the minimal value of $\Delta$ such that there exists a length-$n$ width-$w$ ROBP $P$ such that $P$ computes $\approxcount[n,\Delta]$. Then $\Delta(n,w)\in n/2-\Theta(\sqrt{nw})$.
\end{corollary}

\subsection{Small Error Regime of \tops{$k$}{k}-counter Approximate Counting}

\begin{theorem}\label{thm.algo.small-err}
	For any integers $k\>=2$, $n\>=10k$ and real number $10\<=\Delta\<=n/10$, there exists a length-$n$ width-$w$ ROBP $P$ over alphabet $[k]$ such that $P$ computes $\approxcount\[n,\Delta\]$ for some $w\<=\C{n+k-1-\Omega(\sqrt{n\Delta})}{k-1}$.
\end{theorem}
\begin{Proof}
	Let $l:=\left\lceil\sqrt{\f{n}{\Delta-1}}\right\rceil+1$ and $m:=\lf\sqrt{n(\Delta-1)}\rf-1$.\par
	Suppose the input is $(x_1,\cs,x_n)$. Let $P$ works as follows: in the first to the $\(n-m\)$-th layer we count the exact number of $1,2,\cs,k$'s in $x_1,\cs,x_{n-m}$ via a simple dynamic programming. Suppose $x_1,\cs,x_{n-m}$ contains $a_1,\cs,a_k$ many $1,\cs,k$'s respectively. Now we round $(a_1,\cs,a_k)$ to a new tuple $(b_1,\cs,b_k)$ (which only depends on $a_1,\cs,a_k$) such that $b_j\in\L\lf\f{l-1}l\c a_j\rf,\left\lceil\f{l-1}l\c a_j\right\rceil\R$ and $b_1+\cs+b_k=\lf\f{l-1}l\c\(a_1+\cs+a_k\)\rf=\lf\f{l-1}l\c(n-m)\rf$. Such $b_1,\cs,b_k$ always exists since $\sum_{j=1}^k\lf\f{l-1}l\c a_j\rf\<=\lf\f{l-1}l\c\(a_1+\cs+a_k\)\rf\<=\sum_{j=1}^k\left\lceil\f{l-1}l\c a_j\right\rceil$. In the $\(n-m+1\)$-th layer to the last layer we exactly count the value of $b_j$ plus the number of $j$'s in $x_{n-m+1},x_{n-m+2},\cs$ for each $j\in[k]$ via a simple dynamic programming.\par
	Suppose $x_{n-m+1},\cs,x_n$ contains $c_1,\cs,c_k$ many $1,\cs,k$'s respectively. Then at the last layer we know the value of $(b_j+c_j)$ for each $j\in[k]$. We output $\f{l}{l-1}\(b_j+c_j\)$ as an estimate of the number of $j$'s in $x_1,\cs,x_n$. Note that the exact number of $j$'s in $x_1,\cs,x_n$ is $(a_j+c_j)$ and we have
	\begin{align*}
		\la\f{l}{l-1}\(b_j+c_j\)-(a_j+c_j)\ra&\<=\la\f{l}{l-1}\c b_j-a_j\ra+\f{c_j}{l-1}\<=\f{l}{l-1}+\f{m}{l-1}\<=\Delta.
	\end{align*}
	So $P$ correctly computes $\approxcount\[n,\Delta\]$.\par
	Finally we count how much width $P$ needs. In the $t$-th layer ($t\in[n-m]$), $|V_t|$ is the number of tuples $(a_1,\cs,a_k)\in\mathbb{N}^k$ such that $a_1+\cs+a_k=t$. In the $(n-m+t)$-th layer ($t\in[m]$), $|V_t|$ is the number of tuples $(a_1,\cs,a_k)\in\mathbb{N}^k$ such that $a_1+\cs+a_k=\lf\f{l-1}l\c(n-m)\rf+t$. Therefore the width $w$ of $P$ satisfies
	\begin{align*}
		w&\<=\max\L\C{n-m+k-1}{k-1},\C{\lf\f{l-1}l\c(n-m)\rf+m+k-1}{k-1}\R\\
		&\<=\max\L\C{n-m+k-1}{k-1},\C{n-\left\lceil\f{n-m}l\right\rceil+k-1}{k-1}\R\\
		&\<=\C{n+k-1-\Omega(\sqrt{n\Delta})}{k-1}.
	\end{align*}
\end{Proof}

Combine with \Cref{cor.k-counter-lb.small-err} we get:

\begin{corollary}\label{cor.tight-small-err}
	For any integers $k\>=2$, $n\>=10k$ and real number $10\<=\Delta\<=n/(10k^2)$, let $w(k,n,\Delta)$ be the minimal value of $w$ such that there exists a length-$n$ width-$w$ ROBP $P$ such that $P$ computes $\approxcountc k[n,\Delta]$. Then $w(k,n,\Delta)\in\C{n+k-1-\Theta(\sqrt{n\Delta})}{k-1}$.
\end{corollary}

\section{Implications in other Streaming Lower Bounds}\label{sec.imply}

In this section, we prove streaming lower bounds for other problems. The proofs are based on our lower bound for approximate counting.

\subsection{Lower Bound for Heavy Hitters}

For simplicity of the notions, we define the {\it frequency} of an element $u$ in a list $L$ as the number of $u$'s in $L$.

\begin{definition}\label{def.heavyhitters}
	For any integers $n,U,k$ such that $\min\{n,U\}\gg k\>=2$, $\msf{HeavyHitters}[n,U,k]$ is the following problem: given an input string $\b{x}\in[U]^n$, we are required to output a list of $k$ elements $u_1,\cs,u_k\in[U]$ such that: each element in $[U]$ that appears $\>=n/k$ times in $\b{x}$ is contained in $\{u_1,\cs,u_k\}$. Moreover, let $f_i$ be the frequency of $u_i$ in $\b{x}$, we are also required to output estimates $\wtilde{f_1},\cs,\wtilde{f_k}$ of $f_1,\cs,f_k$ such that $f_i-n/k\<=\wtilde{f_i}\<=f_i$.
\end{definition}

\spreadwords{We recall that the Misra-Gries algorithm \cite{MG82} for $\msf{HeavyHitters}[n,U,k]$ uses }\\$O\(k(\log(n/k)+\log(U/k))\)$ bits of space in the streaming model.

\begin{theorem}\label{thm.lb-heavyhitters}
	\spreadwords{For any integers $n,U,k$ such that $\min\{n,U\}\gg k\>=2$, computing }\\$\msf{HeavyHitters}[n,U,k]$ requires $\Omega\(k(\log(n/k)+\log(U/k))\)$ bits of space in the streaming model.
\end{theorem}
\begin{Proof}
	First we prove that $\msf{HeavyHitters}[n,U,k]$ requires $\Omega\(k\log(U/k)\)$ bits of space.\par
	Consider the following type of inputs: for any subset $I\subseteq[U]$ such that $|I|=\lf n/\lF n/k\rF\rf$, let $\b{y}(I)$ be an input string in $[U]^n$ such that each $u\in I$ appears $\>=\lF n/k\rF$ times. Note that when we read $\b{y}(I)$, the list $\{u_1,\cs,u_k\}$ in our output should contain $I$. Also note that since there are $\binom{U}{\lf n/\lF n/k\rF\rf}$ many subset $I$'s, and each list $\{u_1,\cs,u_k\}$ can only be a valid output for at most $\binom{k}{\lf n/\lF n/k\rF\rf}$ many $\b{y}(I)$'s. So any algorithm that computes $\msf{HeavyHitters}[n,U,k]$ need to have $\>=\binom{U}{\lf n/\lF n/k\rF\rf}/\binom{k}{\lf n/\lF n/k\rF\rf}\>=(U/k)^{\Omega(k)}$ many possible outputs, and therefore requires $\>=\log\((U/k)^{\Omega(k)}\)\>=\Omega(k\log(U/k))$ bits of space.\par
	Then we prove that $\msf{HeavyHitters}[n,U,k]$ requires $\Omega\(k\log(n/k)\)$ bits of space.\par
	Note that any algorithm for $\msf{HeavyHitters}[n,U,k]$ also computes $\approxcountc U[n,n/(2k)]$ on the same input, and the output of $\approxcount$ only depends on the output of $\msf{HeavyHitters}$: for each $u\in[U]$, if $u$ is some $u_i$ in the output list, then $\wtilde{f_i}+n/(2k)$ is an estimate for the frequency of $u$ with an additive error $\<=n/(2k)$; if $u$ is not in the output list, then the frequency of $u$ is $\<=n/k$, and thus $n/(2k)$ is an estimate for the frequency of $u$ with an additive error $\<=n/(2k)$.\par
	Also note that since $U\gg k$, by computing \spreadwords{$\approxcountc U[n,n/(2k)]$, we can also compute $\approxcountc{\lF2k/3\rF}[n,n/(2k)]$, and by \Cref{cor.k-counter-lb.standard}, computing}\\$\approxcountc{\lF2k/3\rF}[n,n/(2k)]$ requires $\Omega(k\log(n/k))$ bits of space in the streaming model. Therefore we conclude that computing $\msf{HeavyHitters}[n,U,k]$ requires $\Omega\(k\log(n/k)\)$ bits of space.\par
	To summarize, computing $\msf{HeavyHitters}[n,U,k]$ requires 
	$$\Omega\(\max\{k\log(n/k),k\log(U/k)\}\)\>=\Omega\(k(\log(n/k)+\log(U/k))\)$$
	bits of space in the streaming model.
\end{Proof}

We remark that the proof above depends on the estimates $\wtilde{f_1},\cs,\wtilde{f_k}$ in the output. If we are only required to output a list $\{u_1,\cs,u_k\}$, we do not know whether we can prove the streaming lower bound. We leave this as an open problem.

\begin{openproblem}\label{open.heavyhitters}
	Prove or disprove: for any integers $n,U,k$ such that $\min\{n,U\}\gg k\>=2$, computing $\msf{HeavyHitters}[n,U,k]$ (without outputting $\wtilde{f_1},\cs,\wtilde{f_k}$) requires $\Omega\(k\log(n/k)\)$ bits of space in the streaming model.
\end{openproblem}

\subsection{Lower Bound for Quantile Sketch}

We first define the {\it rank} of an element $u$ in a list $L$ as the number of elements in $L$ that are at most $u$.
\def\quantile{\text{\sf{Quantile}}}

\begin{definition}\label{def.quantile}
	For any integers $n,U$ and real number $\eps$ such that $\min\{n,U\}\gg 1/\eps\gg1$, let $\quantile[n,U,\eps]$ be the following problem: given a stream of inputs $(x_1,\cs,x_n)\in[U]^n$. For each $1\<=t\<=n$, when a query $x\in[U]$ arrives after we read $x_1,\cs,x_t$, we are required to output an estimate of the rank of $x$ in $x_1,\cs,x_t$, with an additive error $\<=\eps t$.
\end{definition}

We recall that \cite{GSW24}'s algorithm for $\quantile[n,U,\eps]$ uses $O\(\eps^{-1}(\log(\eps U)+\log(\eps n))\)$ bits of space in the streaming model.

\cite{GSW24} presented the following conditional streaming lower bound for quantile sketch.

\begin{theorem}{\rm (\cite{GSW24})}\label{thm.lb-quantile-cond}
	\spreadwords{Assume for any integers $n,k$ such that $n\gg k\gg 1$, computing}\\$\approxcountc k[n,n/k]$ requires $\Omega(k\log(n/k))$ bits of space in the streaming model. Then for any integers $n,U$ and real number $\eps$ such that $\min\{n,U\}\gg 1/\eps\gg1$, computing $\quantile[n,U,\eps]$ requires $\Omega\(\eps^{-1}(\log(\eps U)+\log(\eps n))\)$ bits of space in the streaming model.
\end{theorem}

We remark that \cite{GSW24} stated the assumption in \Cref{thm.lb-quantile-cond} as a conjecture. By resolving their conjecture via our lower bounds for approximate counting, we give an unconditional streaming lower bound for quantile sketch.

\begin{theorem}\label{thm.lb-quantile}
	For any integers $n,U$ and real number $\eps$ such that $\min\{n,U\}\gg 1/\eps\gg1$, computing $\quantile[n,U,\eps]$ requires $\Omega\(\eps^{-1}(\log(\eps U)+\log(\eps n))\)$ bits of space in the streaming model.
\end{theorem}
\begin{Proof}
	By \Cref{thm.lb-quantile-cond}, we only need to prove that: for any integers $n,k$ such that $n\gg k\gg 1$, computing $\approxcountc k[n,n/k]$ requires $\Omega(k\log(n/k))$ bits of space in the streaming model.\par
	\spreadwords{Note that by computing $\approxcountc k[n,n/k]$, we can also compute }\\$\approxcountc{\lF k/3\rF}[n,n/k]$, and by \Cref{cor.k-counter-lb.standard}, computing $\approxcountc{\lF k/3\rF}[n,n/k]$ requires $\Omega(k\log(n/k))$ bits of space. Therefore computing $\approxcountc k[n,n/k]$ requires $\Omega(k\log(n/k))$ bits of space in the streaming model.
\end{Proof}

\section*{Acknowledgement}

The author is thankful to Lijie Chen, Meghal Gupta, Xin Lyu, Jelani Nelson, Mihir Singhal, Avishay Tal, and Hongxun Wu for valuable discussions. The author is also especially thankful to Hongxun Wu for introducing the question of establishing streaming lower bounds for approximate counting and to Lijie Chen, Mihir Singhal, and Hongxun Wu for some useful comments on the early drafts of this paper.

\bibliographystyle{alpha}
\bibliography{ref}

\newcommand{\etalchar}[1]{$^{#1}$}
\begin{thebibliography}{AHNY22}

\bibitem[ABJ{\etalchar{+}}22]{ABJ+22}
Mikl{\'{o}}s Ajtai, Vladimir Braverman, T.~S. Jayram, Sandeep Silwal, Alec Sun,
  David~P. Woodruff, and Samson Zhou.
\newblock The white-box adversarial data stream model.
\newblock In Leonid Libkin and Pablo Barcel{\'{o}}, editors, {\em {PODS} '22:
  International Conference on Management of Data, Philadelphia, PA, USA, June
  12 - 17, 2022}, pages 15--27. {ACM}, 2022.

\bibitem[AHNY22]{AHNY22}
Ishaq Aden{-}Ali, Yanjun Han, Jelani Nelson, and Huacheng Yu.
\newblock On the amortized complexity of approximate counting.
\newblock {\em CoRR}, abs/2211.03917, 2022.

\bibitem[Ajt90]{Ajt90}
Mikl{\'{o}}s Ajtai.
\newblock Approximate counting with uniform constant-depth circuits.
\newblock In Jin{-}Yi Cai, editor, {\em Advances In Computational Complexity
  Theory, Proceedings of a {DIMACS} Workshop, New Jersey, USA, December 3-7,
  1990}, volume~13 of {\em {DIMACS} Series in Discrete Mathematics and
  Theoretical Computer Science}, pages 1--20. {DIMACS/AMS}, 1990.

\bibitem[BGW20]{BGW20}
Mark Braverman, Sumegha Garg, and David~P. Woodruff.
\newblock The coin problem with applications to data streams.
\newblock In Sandy Irani, editor, {\em 61st {IEEE} Annual Symposium on
  Foundations of Computer Science, {FOCS} 2020, Durham, NC, USA, November
  16-19, 2020}, pages 318--329. {IEEE}, 2020.

\bibitem[BGZ21]{BGZ21}
Mark Braverman, Sumegha Garg, and Or~Zamir.
\newblock Tight space complexity of the coin problem.
\newblock In {\em 62nd {IEEE} Annual Symposium on Foundations of Computer
  Science, {FOCS} 2021, Denver, CO, USA, February 7-10, 2022}, pages
  1068--1079. {IEEE}, 2021.

\bibitem[GSW24]{GSW24}
Meghal Gupta, Mihir Singhal, and Hongxun Wu.
\newblock Optimal quantile estimation: beyond the comparison model.
\newblock {\em Electron. Colloquium Comput. Complex.}, pages TR24--065, 2024.

\bibitem[MG82]{MG82}
Jayadev Misra and David Gries.
\newblock Finding repeated elements.
\newblock {\em Sci. Comput. Program.}, 2(2):143--152, 1982.

\bibitem[Mor78]{Mor78}
Robert~H. Morris.
\newblock Counting large numbers of events in small registers.
\newblock {\em Commun. {ACM}}, 21(10):840--842, 1978.

\bibitem[NY22]{NY22}
Jelani Nelson and Huacheng Yu.
\newblock Optimal bounds for approximate counting.
\newblock In Leonid Libkin and Pablo Barcel{\'{o}}, editors, {\em {PODS} '22:
  International Conference on Management of Data, Philadelphia, PA, USA, June
  12 - 17, 2022}, pages 119--127. {ACM}, 2022.

\bibitem[Rou16]{Rou16}
Tim Roughgarden.
\newblock Communication complexity (for algorithm designers).
\newblock {\em Found. Trends Theor. Comput. Sci.}, 11(3-4):217--404, 2016.

\bibitem[RS16]{RS16}
Anup Rao and Makrand Sinha.
\newblock A direct-sum theorem for read-once branching programs.
\newblock In Klaus Jansen, Claire Mathieu, Jos{\'{e}} D.~P. Rolim, and Chris
  Umans, editors, {\em Approximation, Randomization, and Combinatorial
  Optimization. Algorithms and Techniques, {APPROX/RANDOM} 2016, September 7-9,
  2016, Paris, France}, volume~60 of {\em LIPIcs}, pages 44:1--44:15. Schloss
  Dagstuhl - Leibniz-Zentrum f{\"{u}}r Informatik, 2016.

\end{thebibliography}

\appendix

\end{document}